\documentclass[journal]{new-aiaa}
\usepackage[dvipsnames]{xcolor}
\usepackage{environ}

\definecolor{addcol}{RGB}{0,0,255}
\definecolor{delcol}{RGB}{255,0,0}

\newif\ifshowchanges
\showchangesfalse  

\newcommand{\added}[1]{\ifshowchanges\textcolor{addcol}{#1}\else#1\fi}
\newcommand{\deleted}[1]{\ifshowchanges\textcolor{delcol}{#1}\fi}

\NewEnviron{DeletedBlock}{\ifshowchanges{\color{delcol}\BODY}\fi}

\usepackage[utf8]{inputenc}
\usepackage{textcomp}
\usepackage{amsmath}

\usepackage{amsthm}
\newtheorem{theorem}{Theorem}

\usepackage{graphicx}
\usepackage[version=4]{mhchem}
\usepackage{siunitx}
\usepackage{longtable,tabularx}
\usepackage{bm}
\title{\deleted{The} Maximum Initial Mass
       \added{in low-thrust spacecraft trajectory design}}

\author{%
Dario Izzo\footnote{Scientific Coordinator, Advanced Concepts Team;
corresponding author: dario.izzo@esa.int.}
and
Giacomo Acciarini\footnote{Research Fellow, Advanced Concepts Team.}
}

\affil{European Space Agency, Advanced Concepts Team, ESTEC,
       Noordwijk, The Netherlands}
       
\begin{document}

\maketitle

\section*{Introduction}
\lettrine{L}{ow}-thrust optimal control problems are most often formulated either as minimum-propellant (or, equivalently, maximum-final-mass) or minimum-time transfers, and these two problem classes underpin the majority of contemporary work in astrodynamics and spacecraft guidance. 

In interplanetary mission design, time- and fuel-optimal low-thrust problems are routinely employed not only to construct feasible flight trajectories, but also as building blocks for exploring the global structure of the transfer landscape via homotopy methods~\mbox{\cite{zhang2023solution, pan2018new, pan2020bounding, wang2024thrust}}\added{, global searches for locally optimal solutions and time--mass trade
studies \cite{russell2007primer}}, parametric sweeps, and reachability studies \cite{chen2021minimum, acciarini2026reachability, zhang2026, bowerfind2023rapid}.

\deleted{They}\added{Minimum-propellant and minimum-time problems} drive homotopy-based algorithms that connect simple reference solutions to complex, high-fidelity transfers; they parameterize reachability analyses; and they provide training data for surrogate models that emulate optimal guidance policies at negligible \deleted{online} cost. Minimum-time low-thrust problems, in particular, have been the natural backbone for the construction of probability-one homotopy schemes that guarantee convergence to optimal solutions under broad conditions~\cite{pan2018new}, and they have inspired a variety of continuation strategies in thrust, boundary conditions, and complexity in the dynamics~\cite{wang2024thrust, zhang2023solution}.

Despite the pervasive role of low-thrust analysis in preliminary mission design, the structurally related maximum-initial-mass problem has instead received comparatively little explicit attention. Informally, this problem seeks, for a fixed transfer geometry and final state, the largest initial mass from which a low-thrust spacecraft can accomplish the transfer subject to its dynamical and operational constraints. Conceptually, it inverts the usual “given initial mass, minimize time/propellant” viewpoint and instead asks “given time of flight, maximise initial mass”. 
During the Global Trajectory Optimization Competitions (GTOCs), this perspective started to emerge: in one edition, the team from JPL introduced and exploited an approximation to the maximum-initial-mass solution, valid for sufficiently short arcs, that later became known and systematically analyzed under the acronym MIMA~\cite{hennes2016fast}. That approximation proved instrumental in designing a highly complex coordinated multi-rendezvous mission, and it suggests that the underlying optimal-control formulation could also be valuable more generally for mission design.

Subsequent work, including contributions by the present authors, has begun to leverage this idea numerically, even if the problem itself has rarely been formalized and studied on its own terms.
Neural surrogate models of the maximum-initial-mass map have been constructed to rapidly assess the feasibility and performance of large sets of asteroid mining scenarios~\cite{acciarini2024computing,izzo2025asteroid}, while related constructs have been proposed for solar sailing applications and low-thrust reachability envelopes~\cite{acciarini2026reachability}. 
In these studies, the maximum-initial-mass quantity appears as a natural scalar summary of the reachable set for a specified transfer, encoding in a single number the combined effect of thrust, specific impulse, and boundary conditions. Yet, its optimal-control structure, its Hamiltonian, costate behavior, and relation to the classical time-optimal problem, has so far been exposed mostly through numerical evidence rather than a dedicated analytical treatment.

This work argues that the maximum-initial-mass problem deserves to be elevated to a first-class optimal control formulation in low-thrust trajectory design. Building on insights that originated from competitive trajectory design at the GTOC and were later used in the context of preliminary mission analysis, \deleted{the paper}\added{the Note} formally introduces the maximum-initial-mass optimal control problem, characterizes its relationship to the \added{time-of-flight extrema, thus including both the maximum and minimum cases} \deleted{more familiar minimum-time low-thrust problem}, and highlights the specific advantages that follow from adopting this \added{unified} formulation as a design primitive. As a further result, \deleted{we use the derivation of the necessary conditions to} we clarify a common misunderstanding in the indirect formulation of time-optimal problems, specifically regarding the free-time transversality condition. Finally, we show that the maximum-initial-mass problem admits a natural homotopy path to multi-revolution transfers. This path can be exploited to reveal the rich solution landscape that arises in such cases and to provide a fundamentally different perspective from thrust-continuation-based approaches.

By making the maximum-initial-mass formulation explicit and clarifying its links to the \deleted{established minimum-time framework}\added{more familiar time-of-flight extrema}, as well as its differences and advantages, \deleted{the paper}\added{the Note} aims to provide a foundational step toward systematic use of maximum-initial-mass optimal control in low-thrust mission design.

\section*{The maximum initial mass problem}
\noindent
Consider the low-thrust dynamics,
\begin{equation}
\label{eq:dyn}
\dot{\bm{x}} = \bm{f}(\bm{x}) + c_1 \bm{B}(\bm{x}) \frac{u}{m}\hat{\bm{i}},
\qquad
\dot m = -\frac{c_1}{c_2}u,
\end{equation}
where $\bm{x}$ is the state vector, such as Cartesian, modified equinoctial elements, or any other coordinate representation, $u \in [0,1]$ is the throttle, $\hat{\bm{i}}$ is the thrust direction with $\|\hat{\bm{i}}\|=1$, and $c_1$, $c_2$ are constants related to maximum thrust and specific impulse. 
The maximum-initial-mass problem consists in finding the control histories $u(t)$ and $\hat{\bm{i}}(t)$ that transfer the system from $\bm{x}(t_0)=\bm{x}_0$ to $\bm{x}(t_f)=\bm{x_f}$ while maximizing the initial mass $m(t_0)$, with fixed transfer time $t_f-t_0$. 
To investigate its mathematical structure and explain its favorable numerical properties, we begin from the Pontryagin Minimum Principle and derive the corresponding first-order necessary conditions. We formulate the problem in the normal form of the principle, following \cite{longuski2014optimal}, by setting $\lambda_0=1$ (this is the constant that multiplies the objective). This fixes the multiplier normalization and excludes abnormal extremals \cite{longuski2014optimal}. Let us introduce the Mayer-form cost functional, converting maximization into an equivalent minimization,
\begin{equation}
\added{\text{Minimize:}\qquad}
J = -m(t_0),
\end{equation}
where the Hamiltonian is,
\begin{equation}
\mathcal{H}^{MIM}(\bm{x},m,\bm{\lambda},\lambda_m,u,\hat{\bm{i}})
=
\bm{\lambda}^T
\left[
\bm{f}(\bm{x}) + c_1 \bm{B}(\bm{x}) \frac{u}{m}\hat{\bm{i}}
\right]
-\lambda_m \frac{c_1}{c_2}u.
\end{equation}
\added{We normalize the objective multiplier to $\lambda_0=1$ for normal extremals and retain this normalization throughout, except when explicitly discussing abnormal time-extremal representations. In particular, the zero-Hamiltonian case in the converse of Theorem\ref{thm:main} is treated separately and admits a time-extremal representation with $\lambda_0=0$.} 
According to the Pontryagin Maximum Principle in minimum form, the optimal controls satisfy
\begin{equation}
(u^*,\hat{\bm{i}}^*)
=
\operatorname{argmin}_{u \in [0,1],\, \|\hat{\bm{i}}\|=1}
\mathcal{H}^{MIM}(\bm{x},m,\bm{\lambda},\lambda_m,u,\hat{\bm{i}})
\end{equation}
at each instant along the optimal trajectory.
Minimization with respect to the thrust direction gives
\begin{equation}
\label{eq:primer}
\hat{\bm{i}}^*
=
-\frac{\bm{B}(\bm{x})^T\bm{\lambda}}
{\|\bm{B}(\bm{x})^T\bm{\lambda}\|}.
\end{equation}
Substituting this expression into the Hamiltonian yields a reduced Hamiltonian affine in $u$,
\begin{equation}
\mathcal{H}^{MIM}
=
\bm{\lambda}^T\bm{f}(\bm{x})
+
c_1 u
\left(
-\frac{1}{m}\|\bm{B}(\bm{x})^T\bm{\lambda}\|
-
\frac{\lambda_m}{c_2}
\right).
\end{equation}
Therefore, introducing a switching function
\begin{equation}
\label{eq:switching}
\sigma
=
-\frac{1}{m}\|\bm{B}(\bm{x})^T\bm{\lambda}\|
-
\frac{\lambda_m}{c_2},
\end{equation}
the optimal throttle is
\begin{equation}
u^*
=
\begin{cases}
1, & \sigma<0,\\
0, & \sigma>0.
\end{cases}
\end{equation}
To determine the sign of $\sigma$, consider the mass costate equation. Since
$\dot{\lambda}_m = -\partial {\mathcal{H}^{MIM}}/\partial m$, substitution of the optimal thrust direction gives
\begin{equation}
\dot{\lambda}_m
=
\bm{\lambda}^T
\left[
c_1 \bm{B}(\bm{x}) \frac{u}{m^2}\hat{\bm{i}}^*
\right]
=
-\frac{c_1 u}{m^2}\|\bm{B}(\bm{x})^T\bm{\lambda}\|
\le 0.
\end{equation}
Hence $\lambda_m(t)$ is non-increasing along the optimal trajectory. The transversality condition at the initial time is $\lambda_m(t_0)=-\frac{dJ}{dm}=1 > 0$, while at the final time $\lambda_m(t_f)=0$ and therefore $\lambda_m(t)\ge 0$ for all $t \in [t_0,t_f]$. In particular since $m>0$, $c_2>0$, and $\|\bm{B}(\bm{x})^T\bm{\lambda}\|\ge 0$, it follows from Eq.~(\ref{eq:switching}) that in the nondegenerate case, one has in fact $\sigma<0$ and thus, necessarily the optimal throttle is $u^*=1$ and does not exhibit a discontinuous bang-bang structure, making the problem well-suited for numerical solutions. It is indeed the case that both direct and indirect methods converge rapidly to the solution allowing the creation of large databases of optimal trajectories in short times, a property largely utilized in \cite{acciarini2024computing, izzo2025asteroid, acciarini2026reachability}.

Considering indirect methods, a shooting-based method will then guess values for the initial mass and co-states and integrate numerically the augmented dynamics trying to solve for the following shooting function\footnote{The exact details of the shooting function will change depending on the actual boundary manifolds considered for which transversality conditions may apply. Here we show the case of a fixed starting state $\bm{x_0}$ and fixed final state $\bm{x_f}$.}:

\begin{equation}
\label{eq:shooting}
F(\bm\lambda_0, m_0) = [\bm x(t_f)-\bm x_f,  \lambda_m(t_f)]=\bm 0
\end{equation}
which is a well\deleted{-posed} \added{determined} root finding problem, with seven equations and seven unknowns. \added{Finally note that the maximum initial mass problem, after fixing the cost multiplier $\lambda_0=1$, is not homogeneous in the adjoint variables $\boldsymbol \lambda$ since $\lambda_{m}(t_0)=1$.}

\section*{\deleted{Minimum time problem}\added{Time-of-flight extremals}}
\noindent
To later demonstrate the strong relationship of the maximum-initial-mass problem to the \deleted{minimum-time problem}\added{time-of-flight extremals}, we now consider the classical minimum-time low-thrust trajectory optimization problem, under the same dynamics of Eq.~(\ref{eq:dyn}). 
The minimum-time problem consists of finding the control histories $u(t)$ and $\hat{\bm{i}}(t)$ that transfer the system from $\bm{x}(t_0)=\bm{x}_0$ to $\bm{x}(t_f)=\bm{x}_f$ while minimizing the transfer duration $t_f-t_0$, with fixed initial mass $m(t_0)=m_0$. The target state $\bm{x}_f(t_f)$ may be prescribed at a fixed time or, more generally, depend explicitly on $t_f$ and thus be modeling \added{some object to rendezvous with} \deleted {some object to be rendezvous with}.
Introducing the cost functional \added{to be minimized}
\begin{equation}
\label{eq:mintime}
J = t_f - t_0= \int_{t_0}^{t_f} 1 \, dt,
\end{equation}
the Hamiltonian is
\begin{equation}
\mathcal{H}^{tmin}(\bm{x},m,\bm{\lambda},\lambda_m,u,\hat{\bm{i}})
 =
\bm{\lambda}^T
\left[
\bm{f}(\bm{x}) + c_1 \bm{B}(\bm{x}) \frac{u}{m}\hat{\bm{i}}
\right]
-\lambda_m \frac{c_1}{c_2}u + 1 = \mathcal{H}^{MIM} + 1.
\end{equation}
and the optimal $\hat{\bm{i}}^*$ is identical to that presented in Eq.~(\ref{eq:primer}) as well as the optimal throttle $u^*=1$. \added{Note that also for this case we have assumed the scaling of the objective function $\lambda_0 = 1$.} The augmented equations are also the same since the Hamiltonian has the same partial derivatives w.r.t. state and co-states. \added{Since the final time is now left free, we need to account for an added transversality condition, which in this case can be written as $\mathcal H_f^{tmin} < 1$, the equality being excluded since we are not considering abnormal cases.}

Differently from the maximum initial mass, in the minimum-time indirect formulation, the adjoint variables are only determined up to multiplication by a strictly positive scalar, so different initial costate vectors that differ only by a scale produce the same control law and the same state trajectory. Because of that scaling invariance, the unknowns in the shooting problem are not unique unless one extra normalization condition is imposed. 
This additional condition is often imposed \added{through forcing a suitable value for the terminal Hamiltonian (e.g. $\mathcal H_f^{tmin}=0$)} and is frequently presented as a necessary terminal condition in later treatments of minimum-time low-thrust problems~\cite{wang2024thrust,taheri2017co,shen1998time}.
However, Pontryagin’s original treatment of time-optimality suggests the use of scale-free terminal inequality, e.g. $\mathcal H_f^{tmin} < 1$, rather than the more common equality condition~\cite{pontryagin}. 
Accordingly, in a shooting formulation any \added{compatible} gauge choice can be used to remove the positive scaling indeterminacy of the adjoint vector, including none. \deleted{this point is discussed in the following remark.}
For indirect methods, a shooting-based method \deleted{will}\added{can} then guess values for the initial co-states and integrate numerically the augmented dynamics trying to solve for the shooting function:

\begin{equation}
\label{eq:shooting_time}
F(\bm\lambda_0, \lambda_{m_0}, t_f) = [\bm x(t_f)-\bm x_f, \lambda_m(t_f)]=\bm 0\text{.}
\end{equation}
This leads to a root-finding problem with eight unknowns and seven equations. \deleted{so one additional condition is required to render the shooting formulation well-posed.} \added{After convergence one may check the final Hamiltonian value: $\mathcal H_f^{tmin} < 1$ will indicate a minimum time. Which raises the question on what are solutions where $\mathcal H_f^{tmin} > 1$. By revisiting the above developments changing the sign in Eq.~(\ref{eq:mintime}) and thus deriving the Pontryagin conditions for a maximum time, it is easy to discover that those normal extremals have the necessary condition $\mathcal H_f^{tmin} > 1$ while the others are unchanged. Solving Eq.~(\ref{eq:shooting_time})} one then will converge to a time extremal, whether it corresponds to the maximum-time or minimum-time formulation being determined by the final value of the Hamiltonian. If one would want to select the branch adding the appropriate gauge condition would guarantee the branch selection: e.g. $\mathcal H_f^{tmin}=0$ for minimum time, $\mathcal H_f^{tmin}=2$ for maximum time.

\deleted{
\paragraph{Remark on the free-time transversality condition.}
For free-time problems with a moving end point, the classical transversality condition
\begin{equation}
\label{eq:classic_transversality}
\mathcal H(t_f)+\bm\lambda(t_f)^T\dot{\bm x}_f(t_f)=0
\end{equation}
is commonly derived under a variation argument for the terminal time when the terminal point is constrained to lie on a moving target \cite{liberzon2012calculus}. For minimum-time problems, however, the optimal transfer duration is attained at the boundary of the feasible transfer-time set, so the applicability of that argument requires careful discussion.
The condition in Eq.~(\ref{eq:classic_transversality}) is derived under the assumption that the terminal time is an interior admissible variable, so that two-sided first-order variations with respect to $t_f$ are allowed. In reality, however, for minimum time problems, the admissible terminal time belongs to a half-line $t_f \in [t_{f_{\min}}, \infty)$ and the optimum is trivially attained at its lower boundary.
Hence, the relevant first-order condition is one-sided and of Karush-Kuhn-Tucker (KKT) type \cite{nocedal2006numerical} and reads as
$
\mathcal H(t_f)+\bm\lambda(t_f)^T\dot{\bm x}_f(t_f) < 0.
$
To formalize the above argument we go through a possible derivation of the transversality condition valid for time-free optimal control problems with moving endpoint:
$
\mathcal H(t_f)+\bm\lambda(t_f)^T\dot{\bm x}_f(t_f)=0.
$
We proceed by augmenting the state space with the physical time, introducing an additional state variable $s(t)=t,\quad \dot s = 1,$ and then introducing the normalized time $\tau \in [0,1]$ through $t=t_0+T\tau$, where $T=t_f-t_0$. 
In the minimum-time setting, the cost thus becomes $J=T$, while the augmented dynamics becomes,
$$
\frac{d\bm{x}}{d\tau}
=
T\left(\bm{f}(\bm{x}) + c_1\bm{B}(\bm{x})\frac{u}{m}\hat{\bm{i}}\right),
\qquad
\frac{dm}{d\tau}
=
-T\frac{c_1}{c_2}u,
\qquad
\frac{ds}{d\tau}=T.
$$
The terminal condition $\bm x(t_f)=\bm x_f(t_f)$ is then rewritten in the augmented state space as a terminal fixed manifold condition,
$
\mathcal M_f
=
\left\{
(\bm x,m,s)\;:\;\bm x=\bm x_f(s)
\right\}.
$
Introducing the additional costate $\lambda_s$ associated with $s$, the corresponding augmented Hamiltonian is
$
\mathcal H_a=\mathcal H+\lambda_s,
$
and the corresponding rescaled Hamiltonian is simply $\widetilde{\mathcal H}=T\mathcal H_a$. 
A formal stationarity argument with respect to $T$ would then yield
$
\frac{\partial \widetilde{\mathcal H}}{\partial T}=\mathcal H_a=0.
$
Using the transversality condition at the terminal manifold $\mathcal M$, one has
$
\lambda_s(t_f)=-\bm\lambda(t_f)^T\dot{\bm x}_f(t_f),
$
up to sign convention, so that the previous condition $\mathcal H_a=0$ becomes precisely
$
\mathcal H(t_f)+\bm\lambda(t_f)^T\dot{\bm x}_f(t_f)=0.
$
However, this step is not valid in a minimum-time problem. 
Indeed, after normalization, $T$ is minimized over the admissible set $[T_{\min},\infty)$, and the optimum is attained at the left boundary, namely $T=T_{\min}$, since we are minimizing time. Hence, the minimization with respect to $T$ is a boundary optimum, not an interior one, and the appropriate first-order condition is of KKT type.
Therefore the use of $\mathcal H_a(t_f)=0$ as an independent necessary terminal condition for the associated indirect shooting formulation is not needed and can be substituted by any other gauge condition on the co-states. This observation unlocks, for minimum time problems, the study of alternative gauge conditions to improve the numerical stability of the shooting method employed, or facilitate homotopy approaches. 
The statement that $\mathcal H(t_f)+\bm\lambda(t_f)^T\dot{\bm x}_f(t_f)=0$ is not a necessary condition for time optimal problems may appear to be in tension with some of the modern works utilizing the maximum principle, on the other hand the original work from Pontryagin \cite{pontryagin} is actually clear about it when deriving in Chapter I, (\S 3 \lq\lq The maximum principle\rq\rq\ ) the necessary condition for time optimality as $\mathcal H^*(t_f) \le 0$, and not $\mathcal H^*(t_f) = 0$\footnote{In reporting the original result from Pontryagin, we have here adapted the notation to the one used in this work as well as changed the sign as we are dealing with the minimum principle, not the maximum}. Also in the work \cite{lu2008rapid} Lu and colleagues seem to have noted the issue and write about a \lq\lq common misconception\rq\rq\ in a related discussion. Finally, as noted by Longuski, Guzm{\'a}n and Prussing \cite{longuski2014optimal}, the correct proof of the principle is notoriously difficult and often either avoided by scholars or highly simplified. It is then not surprising the subtle point discussed here, differentiating free-time from minimum-time, eluded practitioners considering also that minimum time problems in any case require an additional gauge condition when solved by shooting and $\mathcal H_a(t_f) = 0$ happens to be a perfectly valid one (just not a necessary one).
}

\added{
\paragraph{Remark on time extremals and Hamiltonian normalization.}
Throughout this remark only normal time extremals are considered. 
The minimum-time and maximum-time problems, written respectively
as the minimization of $t_f-t_0$ and $-(t_f-t_0)$, have Hamiltonians
$$
\mathcal{H}^{tmin}=\mathcal{H}^{MIM}+1,
\qquad
\mathcal{H}^{tmax}=\mathcal{H}^{MIM}-1.
$$
They therefore share the same state and adjoint equations, Hamiltonian-minimizing controls, and terminal mass condition $\lambda_m(t_f)=0$.
Consider first prescribed, time-independent endpoint states.
Free-time transversality requires
$$
\mathcal{H}^{MIM}(t_f)=-1
\quad\text{for minimum time},
\qquad
\mathcal{H}^{MIM}(t_f)=+1
\quad\text{for maximum time}.
$$
Positive rescaling of the state costates preserves the common state-costate equations, the minimizing controls, and the homogeneous endpoint conditions, while multiplying $\mathcal{H}^{MIM}$ by the same positive factor. 
Consequently, a costate representation satisfying these common conditions can be positively rescaled to satisfy minimum-time transversality if and only if $\mathcal{H}^{MIM}(t_f)<0$, and maximum-time transversality if and only if
$\mathcal{H}^{MIM}(t_f)>0$. 
To express both branches using the same quantity we consider $\mathcal{H}^{tmin}$.
The scale-invariant conditions are then
$$
\mathcal{H}^{tmin}(t_f) < 1
\quad\text{for the minimum-time branch},
\qquad
\mathcal{H}^{tmin}(t_f) > 1
\quad\text{for the maximum-time branch}.
$$
Here $\mathcal{H}^{tmin}$ is used as a common diagnostic quantity for both branches and
should not be confused with $\mathcal{H}^{tmax}$. 
These inequalities reconcile directly with the terminal equality $\mathcal{H}^{tmin}(t_f)=0$ commonly imposed in minimum-time formulations. 
Indeed, under the positive rescaling
$$
(\boldsymbol{\lambda},\lambda_m)
\longmapsto
b(\boldsymbol{\lambda},\lambda_m),
\qquad b>0,
$$
one has
$$
\mathcal{H}_b^{tmin}(t_f)
=
1+b\bigl(\mathcal{H}^{tmin}(t_f)-1\bigr).
$$
If $\mathcal{H}^{tmin}(t_f) < 1$, the unique choice
$$
b=\frac{1}{1-\mathcal{H}^{tmin}(t_f)}
$$
gives $\mathcal{H}_b^{tmin}(t_f)=0$, without changing the state-control
trajectory. Similarly, if $\mathcal{H}^{tmin}(t_f) > 1$, the choice
$$
b=\frac{1}{\mathcal{H}^{tmin}(t_f)-1}
$$
gives $\mathcal{H}_b^{tmax}(t_f)=0$. Thus, the terminal equality fixes the costate scale required
by the objective-normalized free-time Pontryagin conditions.
An equivalent shooting formulation may therefore use another
costate normalization, such as $\lambda_m(t_0)=1$.
The costates used in that formulation need not themselves
satisfy the terminal equality: they admit a positively rescaled
representation that does.
Unlike the equality, this inequality does not depend on a positive scaling of the costates, which is discussed by Lu et al.\cite{lu2008rapid} as the reason one can eventually drop the constraint $\mathcal{H}(t_f)=0$ \cite{lu2008rapid}. Ma et al.~\cite{ma2021practical} seem to also have noted this fact and replace in their solver the classical constraint by $\mathcal{H}^{tmin}(t_f) < 1$, reporting an enlarged convergence domain.
Finally, as noted by Longuski, Guzm{\'a}n and Prussing~\cite{longuski2014optimal}, the correct proof of the principle is notoriously difficult and often avoided or simplified, so it is then not surprising the subtle point discussed here, valid only for time-extremal
problems, eludes practitioners, considering also that minimum time problems in any case require an additional gauge condition when solved by shooting and the use of $\mathcal H^{tmin}(t_f) = 0$ happens to be a perfectly valid option, just not the only one.}

\added{
It is interesting to treat here briefly also the case of a moving target $\boldsymbol{x}_f(t_f)$, where free-time transversality 
\begin{equation}
\label{eq:classic_transversality}
\mathcal H^{tmin}(t_f)-\bm\lambda(t_f)^T\dot{\bm x}_f(t_f)=0
\end{equation}
includes the displacement of the terminal state induced by a change in arrival time (see \cite{bryson2018applied} \S2.7). 
Let us define
$$
q_f =
\mathcal H^{MIM}(t_f)
-\boldsymbol{\lambda}(t_f)^T
\dot{\boldsymbol{x}}_f(t_f),
$$
the transversality conditions become
$$
q_f=-1
\quad\text{for minimum time},
\qquad
q_f=+1
\quad\text{for maximum time}.
$$
Since $q_f$ scales linearly with the state costates, the preceding
argument applies with $\mathcal H^{MIM}(t_f)$ replaced by $q_f$. In particular, the scale-invariant conditions for the normal minimum-time and maximum-time branches are $q_f<0$ and
$q_f>0$, respectively, rather than conditions on the sign of $\mathcal H^{MIM}(t_f)$ alone. 
An important specialization occurs when the target follows
the same uncontrolled dynamics as the spacecraft,
$$
\dot{\boldsymbol{x}}_f(t)
=
\boldsymbol{f}(\boldsymbol{x}_f(t)).
$$
At rendezvous, $\boldsymbol{x}(t_f)=\boldsymbol{x}_f(t_f)$.
The drift contribution to $\mathcal H^{MIM}(t_f)$ is therefore
$\boldsymbol{\lambda}(t_f)^T
\boldsymbol{f}(\boldsymbol{x}_f(t_f))$, exactly equal to the
target-motion term subtracted in $q_f$. These two terms cancel,
leaving
$$
q_f =
\frac{c_1u_f}{m_f}
\boldsymbol{\lambda}(t_f)^T
\boldsymbol{B}(\boldsymbol{x}_f(t_f))
\hat{\boldsymbol{i}}_f
-\lambda_m(t_f)\frac{c_1}{c_2}u_f.
$$
Using $\lambda_m(t_f)=0$ and the Hamiltonian-minimizing thrust
direction yields
$$
q_f =
-\frac{c_1u_f}{m_f}
\left\|
\boldsymbol{B}(\boldsymbol{x}_f(t_f))^T
\boldsymbol{\lambda}(t_f)
\right\|
\le 0,
$$
Hence the normal maximum-time branch is impossible in this
case. Finally, the sign conditions identify the time formulation for which the rescaled multipliers satisfy the Pontryagin necessary conditions; they do not, by themselves, establish local or global optimality.
}

The above remark turns out to be useful for proving the following \deleted{T}\added{t}heorem, a central result of this work.

\added{
\begin{theorem}[Correspondence between time extremals and maximum-initial-mass extremals]
\label{thm:main}
Consider the dynamics in Eq.~(1), with prescribed endpoint states $\bm{x}(t_0)=\bm{x}_0$ and $\bm{x}(t_f)=\bm{x}_f$, and free terminal mass. The following statements hold.
\begin{enumerate}
\item \emph{Time extremals to maximum-initial-mass extremals.} Let $(\bm{x}^*(\cdot),m^*(\cdot),u^*(\cdot),\hat{\bm{i}}^*(\cdot),t_f^*)$ be an extremal of a minimum-time or maximum-time problem with prescribed initial mass $m(t_0)=m_0$. Then $(\bm{x}^*(\cdot),m^*(\cdot),u^*(\cdot),\hat{\bm{i}}^*(\cdot))$ is an extremal of the maximum-initial-mass problem with prescribed terminal time $t_f=t_f^*$.
\item \emph{Maximum-initial-mass extremals to time extremals.} Let $(\bm{x}^*(\cdot),m^*(\cdot),u^*(\cdot),\hat{\bm{i}}^*(\cdot))$ be an extremal of the maximum-initial-mass problem with prescribed terminal time $t_f=t_f^*$. Then $(\bm{x}^*(\cdot),m^*(\cdot),u^*(\cdot),\hat{\bm{i}}^*(\cdot),t_f^*)$ is an extremal of either the minimum-time or the maximum-time problem with prescribed initial mass $m(t_0)=m^*(t_0)$.
\end{enumerate}
\end{theorem}
}
\deleted{
\noindent
Under the dynamics in Eq.~(\ref{eq:dyn}) let
$
(\bm{x}^*(\cdot),m^*(\cdot),u^*(\cdot),\hat{\bm{i}}^*(\cdot),t_f^*)
$
be an extremal solution of the minimum-time problem from the prescribed initial condition
$
\bm{x}(t_0)=\bm{x}_0,\qquad m(t_0)=m_0,
$
to the target condition
$
\bm{x}(t_f)=\bm{x}_f.
$
Then the restricted trajectory-control pair
$
(\bm{x}^*(\cdot),m^*(\cdot),u^*(\cdot),\hat{\bm{i}}^*(\cdot))
$
is also an extremal of the maximum-initial-mass problem over the fixed time interval $[t_0,t_f^*]$.
Conversely, if
$
(\bm{x}^*(\cdot),m^*(\cdot),u^*(\cdot),\hat{\bm{i}}^*(\cdot))
$
is an extremal of the maximum-initial-mass problem over the fixed time interval $[t_0,t_f^*]$, then, with
$
m_0 = m^*(t_0),
$
the augmented tuple
$
(\bm{x}^*(\cdot),m^*(\cdot),u^*(\cdot),\hat{\bm{i}}^*(\cdot),t_f^*)
$
is an extremal of the corresponding minimum-time problem.
}
\begin{proof}
\deleted{Since}\added{Let} $(\bm{x}^*(\cdot),m^*(\cdot),u^*(\cdot),\hat{\bm{i}}^*(\cdot),t_f^*)$ \deleted{is}\added{be} a \added{time} extremal \deleted{solution of the minimum-time problem}, \added{then} it satisfies the state equations, the adjoint equations, the minimum condition on the Hamiltonian, and the terminal constraints (see Table \ref{tab:mt_mim}, left column). In particular, since the final mass is free \deleted{in the minimum-time}\added{when looking for a time extremal} \deleted{formulation}, the \added{corresponding} transversality condition gives
$
\lambda_m(t_f)=0.
$
Now as observed, the co-state vector of the minimum-time problem is defined only up to multiplication by a \deleted{nonzero constant} \added{strictly positive constant}, and this scaling does not alter the associated state trajectory nor the control law obtained from the minimum condition. By the positivity of $\lambda_m(t_0)$ established earlier, one may positively rescale the costates so that $\lambda_m(t_0)=1$.
With this choice, the same trajectory-control pair satisfies the Pontryagin necessary conditions for the maximum-initial-mass problem on the fixed interval $[t_0,t_f^*]$: the dynamics are unchanged, the Hamiltonian minimization with respect to $u$ and $\hat{\bm{i}}$ is unchanged, the final time is fixed, and the free-initial-mass transversality condition is precisely
$
\lambda_m(t_0)=1.
$
Hence the same trajectory-control pair defines an extremal also for the maximum-initial-mass problem. 

Conversely, let
$
(\bm{x}^*(\cdot),m^*(\cdot),u^*(\cdot),\hat{\bm{i}}^*(\cdot))
$
be an extremal of the maximum-initial-mass problem over the fixed time interval $[t_0,t_f^*]$. Then there exists a corresponding set of adjoint variables such that the state equations, the adjoint equations, the minimum condition on the Hamiltonian, and the terminal transversality condition
$
\lambda_m(t_f^*)=0
$
are satisfied (see Table \ref{tab:mt_mim}, right column). Fixing the initial mass to
$
m_0=m^*(t_0),
$
the same trajectory-control pair, together with the same terminal time $t_f^*$, satisfies the necessary conditions for \deleted{the corresponding minimum-time problem}\added{it to be a time extremal}: the dynamics are unchanged, the adjoint equations are unchanged, the Hamiltonian minimization with respect to $u$ and $\hat{\bm{i}}$ is unchanged, the initial mass is now prescribed, and the final mass remains free, yielding again
$
\lambda_m(t_f^*)=0.
$
\added{Finally, if $\mathcal H^{MIM}(t_f^*)<0$, a positive rescaling of the
costates yields a normal minimum-time extremal. If
$\mathcal H^{MIM}(t_f^*)>0$, it yields a normal maximum-time
extremal. If $\mathcal H^{MIM}(t_f^*)=0$, the same trajectory and costates satisfy the abnormal Pontryagin conditions for both time formulations, corresponding to trajectories that are simultaneously minimum-time extremals and maximum-time extremals.}
Therefore
$
(\bm{x}^*(\cdot),m^*(\cdot),u^*(\cdot),\hat{\bm{i}}^*(\cdot),t_f^*)
$
is a\deleted{n} \added{time} extremal\deleted{ of the corresponding minimum-time problem}.
\end{proof}

Note that the terminal equality $\mathcal H(t_f)=0$ is not imposed as an additional condition to prove the \added{first part of the} theorem.
Imposing it for the \deleted{minimum-}time \added{extremals} \deleted{problem} would, in general, select a \deleted{different}\added{specific} scaling of the adjoint variables and thus lead to an initial value of \(\lambda_m\) different from \(1\), thereby invalidating the proof. 
Since, as discussed above, the condition \(\mathcal{H}(t_f) = 0\) \deleted{does not apply} \added{is not necessary} to \deleted{minimum-}\added{time extremals} \deleted{problems}, \deleted{the adjoint vector is defined only up to a nonzero scalar factor. 
W}\added{w}e can \deleted{therefore} normalize it by prescribing \(\lambda_m(t_0) = 1\) \added{as gauge condition}, and this normalization is precisely the one that allows one to conclude the proof on the extremal correspondence.

\begin{table}[hbt!]
\caption{\label{tab:mt_mim} Comparison between time and maximum-initial-mass normal extremals. \deleted{, once the gauge condition $\lambda_m(t_0)=1$ is chosen for the minimum-time problem.}}
\centering
\begin{tabular}{lcc}
\hline
\added{Item} & Time & Maximum initial mass (final time fixed)\\
\hline
Cost functional
&
$\displaystyle 
J=\pm\int_{t_0}^{t_f} 1\,dt
$
&
$\displaystyle
J=-m(t_0)
$
\\[0.8em]

Hamiltonian
&
$\displaystyle
\bm{\lambda}^T
\left[
\bm{f}(\bm{x}) + c_1 \bm{B}(\bm{x}) \frac{u}{m}\hat{\bm{i}}
\right]
-\lambda_m \frac{c_1}{c_2}u \pm 1
$
&
$\displaystyle
\bm{\lambda}^T
\left[
\bm{f}(\bm{x}) + c_1 \bm{B}(\bm{x}) \frac{u}{m}\hat{\bm{i}}
\right]
-\lambda_m \frac{c_1}{c_2}u
$
\\[1.2em]
Optimal controls
&
$u^*=1$, $\hat{\bm{i}}^*
=
-\frac{\bm{B}(\bm{x})^T\bm{\lambda}}
{\|\bm{B}(\bm{x})^T\bm{\lambda}\|}$
&
$u^*=1$, $\hat{\bm{i}}^*
=
-\frac{\bm{B}(\bm{x})^T\bm{\lambda}}
{\|\bm{B}(\bm{x})^T\bm{\lambda}\|}$\\[0.4em]

Prescribed initial conditions
&
$\bm{x}(t_0) = \bm{x}_0, \;\; m(t_0) = m_0$
&
$\bm{x}(t_0) = \bm{x}_0,\;\; \lambda_m(t_0)=1$
\\[0.4em]

Prescribed final conditions
&
$\bm{x}(t_f) = \bm{x}_f,\;\; \lambda_m(t_f)=0$
&
$\bm{x}(t_f) = \bm{x}_f,\;\; \lambda_m(t_f)=0$\\[0.4em]

Free variables
&
$\bm{\lambda}(t_0),\;t_f$
&
$\bm{\lambda}(t_0),\;m(t_0)$
\\[0.4em]
Free time transversality
&
$\mathcal H^{tmin}(t_f) < 1$, minimum time
&
none
\\[0.4em]
&
$\mathcal H^{tmin}(t_f) > 1$, maximum time
&
\\
\hline
\end{tabular}
\end{table}

\section*{Homotopy and maximum-initial-mass: the multiple revolution, time-optimal case}
Theorem~\ref{thm:main} lays bare the correspondence between the 
maximum-initial-mass and time optimal problems. 
Table~\ref{tab:mt_mim} collects the relevant conditions dictated by Pontryagin's maximum principle for both formulations and makes their structural similarity explicit. 
\deleted{In the time formulation, a gauge condition must be imposed to obtain a well-posed  shooting problem; in the proof we adopt the normalization $\lambda_m(t_0)=1$, but the correspondence itself is independent of this particular choice. }
In other words any minimum-time or maximum-time extremal with prescribed initial mass $m_0$ and terminal time $t_f^*$ also defines a maximum-initial-mass extremal with prescribed terminal time
$t_f^*$. Conversely, any maximum-initial-mass extremal with initial mass $m_0^*$ and nonzero terminal dynamics-only Hamiltonian defines a minimum-time or maximum-time extremal with prescribed initial mass $m_0^*$.
This correspondence concerns Pontryagin necessary conditions and does not, by itself, establish local or global optimality.
\deleted{The maximum-initial-mass and optimal-time formulations are closely related, but they are not the same. The essential distinction lies in the scalar quantity treated as free: the maximum-initial-mass problem fixes the transfer time and frees the initial mass, whereas the minimum-time problem fixes the initial mass and frees the final time.}
\added{The maximum-initial-mass and time optimal formulations are closely related but remain distinct: the former fixes the transfer time and treats the initial mass as a free variable, whereas the latter fixes the initial mass and treats the final time as free. Despite their close relationship, the two problems \deleted{exhibit different numerical properties}\added{are different problems} provid\added{ing}\deleted{e} complementary viewpoints on \deleted{fundamental} low-thrust trajector\deleted{y}\added{ies} \deleted{problems}, including reachability analysis and homotopy techniques. The maximum-initial-mass formulation characterizes reachability through a single mass threshold, $m_0^*$, rather than through two limiting transfer times, $t_{\min}$ and $t_{\max}$, defining the interval $t_f \in [t_{\min},t_{\max}]$.}
The added value of the maximum-initial-mass formulation for low-thrust reachability analysis has been preliminarily discussed and exploited in recent work \cite{acciarini2026reachability}; here, we focus instead on possible new perspectives added to homotopy research by showing how, \added{in the case here considered}, it allows one to construct a smooth and well-behaved homotopy path for multiple-revolution minimum-time problems. 
In these problems, it is well known how as the available thrust decreases, an increasing number of revolutions are needed, and a large number of locally optimal solutions appear. 
A classical way to navigate this increasingly complex solution landscape is the use of homotopy, with thrust continuation along minimum-time solutions being among the most widely used approaches. 
However, such continuation leads to ill-conditioned paths, which must be regularized, for instance by introducing a second layer of continuation~\cite{wang2024thrust}. 
\deleted{We show, through a simple example, that the maximum-initial-mass problem enables one to define a differentiable and well-conditioned homotopy path leading efficiently to the global time-optimal solution, while also providing a clear view of the overall funnel-shaped structure of the minimum-time multiple-revolution landscape.}
\added{Through a simple example, we show that the maximum-initial-mass problem yields a smooth homotopy path that, without any regularization, visits the locally time-optimal solutions on the explored portion of the iso-$\mathcal{M}$ curve, and reveals their funnel-like arrangement.}

\begin{table}[tb]
\centering
\caption{Computed extremal transfers for $T_{\max}=3\,\mathrm{N}$ and $m_0=1500\,\mathrm{kg}$. \label{tab:alltraj}}
\begin{tabular}{lccc}
\hline
Case & Transfer time (hours) & Final mass (kg) & $N_{\mathrm{rev}}$ \\
\hline
S & 308.34   &1329.77 & 14.35 \\
A \cite{caillau20033d, wang2024thrust} & 283.32   & 1343.58 & 14.66 \\
B \cite{wang2024thrust} & 281.97   & 1344.33 & 15.16 \\
C & 283.78   & 1343.33 & 16.16 \\
D  & 284.16  & 1343.12 & 14.84 \\
E \cite{wang2024thrust} & 285.77  & 1342.23 & 15.83 \\
\hline
\end{tabular}
\end{table}

\added{
Two formulations are used in our homotopy construction and should not be confused.
(i) The maximum-initial-mass problem with free transfer time and fixed terminal true longitude $L_f$, in which $t_f$ is an unknown determined together with the trajectory. 
In this case, the classical free-time transversality condition
$\mathcal H^{MIM}(t_f)=0$ is applied.
(ii) The maximum-initial-mass problem with fixed transfer time $t_f$ and fixed terminal true longitude $L_f$, which is the one appearing in Theorem~\ref{thm:main}. Here $t_f$ is a prescribed parameter and $\mathcal H^{MIM}$ is merely constant along the extremal.
}

We consider the benchmark transfer originally studied by Caillau et al.~\cite{caillau20033d} and subsequently revisited by Wang~\cite{wang2024thrust}. 
A spacecraft of initial wet mass $m_0 = 1500~\mathrm{kg}$ is equipped with a continuous low-thrust propulsion system and is required to perform a multiple-revolution transfer from a geostationary transfer orbit (GTO) to a geostationary orbit (GEO).
The initial orbit is characterized by the Keplerian elements $a_0 = 26\,571\,430~\mathrm{m}$, $e_0 = 0.75$, $i_0 = 7^\circ$, and true anomaly $\nu_0 = 180^\circ$, while the target orbit satisfies $a_f = 42\,165\,000~\mathrm{m}$, $e_f = 0$, and $i_f = 0$. 
The thrust magnitude is assumed constant and bounded by $T_{\max} = c_1 = 3~\mathrm{N}$. \added{The specific impulse is $I_{\mathrm{sp}}=1994.75\,\mathrm{s}$, with $c_2=g_0I_{\mathrm{sp}}$ and $g_0=9.80665\,\mathrm{m\,s^{-2}}$.}
The system state $\bm{x} \in \mathbb{R}^6$ is expressed in terms of modified equinoctial elements~\cite{walker1985set}, we denote with $L$ the true longitude. 
For prescribed values of the final time $t_f$ and terminal number of revolutions $N_{rev} = \frac{L_f}{2\pi}$, the shooting function introduced in Eq.~(\ref{eq:shooting}) implicitly defines the maximum initial mass $m_0$ as a function of these parameters, i.e.,
$$
\mathcal{M} : (t_f, N_{rev}) \mapsto m_0(t_f, N_{rev}).
$$
Under this formulation, the level sets of $\mathcal{M}$, i.e., the iso-$\mathcal{M}$ curves in the $(t_f,  N_{rev})$-plane, correspond to families of extremal trajectories achieving identical initial mass \added{and the 1500 kg contour represents the computed family of maximum-initial-mass extremals having initial mass 1500 kg.}
\added{We solve problems in single shooting using modified Powell method with tolerance $\epsilon_P = 10^{-10}$, while the numerical integration is performed using a Taylor adaptive numerical scheme with absolute tolerance $\epsilon_T = 10^{-13}$.}\footnote{More details on this experiment, can be found online: \url{https://esa.github.io/pykep/notebooks/maximum_initial_mass.html}}

We initialize the continuation procedure, in analogy with thrust homotopy approaches, by constructing a low-revolution reference trajectory to serve as a seed solution. 
Specifically, we utilize $N_{\mathrm{rev}} = 1.59$ and \deleted{$t_f = 285.77~\mathrm{h}$}\added{$t_f = 16.66~\mathrm{h}$}, corresponding to a maximum admissible initial mass $\mathcal{M}(t_f, N_{\mathrm{rev}}) = 88.61~\mathrm{kg}$. \added{This trajectory, found using formulation (i), serves only as an auxiliary seed for the continuation and is not intended to satisfy the prescribed initial mass of 1500 kg. 
It was chosen as an easy, low-revolution problem; its low initial mass ($88.61$~kg) is a consequence of this choice, not a design goal, and the prescribed mass is reached through continuation.
A different seed would only change the point at which the $\mathcal M = 1500$~kg iso-curve is first reached.
The iso-curve itself, which we construct and visualize in Fig.~1, is a property of the problem and does not depend on the seed. 
Since it was found to be connected in the explored region, any seed reaching it gives access to the same curve, and hence to the same results.}
Starting from this solution, we perform a continuation on the parameter $N_{\mathrm{rev}}$, progressively increasing the number of revolutions and, consequently, the maximum initial mass. \deleted{During this phase, the final time $t_f$ is treated as a free variable and is determined as part of the shooting procedure.}
\deleted{Accordingly, the free-time transversality condition on the Hamiltonian is enforced. This does not contradict the earlier remark regarding the lower bound on the final time. Indeed, the present formulation considers a maximum initial mass problem with free final time.}
\added{During this first phase, $N_{\mathrm{rev}}$ is increased while $t_f$ is a free variable determined by the shooting procedure, i.e., we solve formulation~(i), and the condition $\mathcal H^{MIM}(t_f)=0$ is therefore enforced. 
This phase is used only to construct a path from the low-mass seed to the target value $\mathcal M=1500$~kg. 
Once this level is reached, we continue along the iso-$\mathcal M$ curve with $t_f$ fixed at each point, i.e., with formulation~(ii), so that Theorem~\ref{thm:main} applies and each point is also an extremal of the minimum-time problem.}
For the case considered here, the resulting homotopy is observed to be smooth and the continuation path remains well-posed, allowing the solution to be reliably tracked up to the prescribed value of the initial mass.
The \deleted{final} trajectory \added{(denoted by $S$)} obtained through this process \deleted{(denoted by $S$)} is characterized by $t_f = 308.34~\mathrm{hours}$ and $N_{\mathrm{rev}} = 14.35$, with $\mathcal{M} = 1500~\mathrm{kg}$.
From this point onward, one can continue along the iso-$\mathcal{M}$ curve \added{using formulation~(ii)}. By Theorem~\ref{thm:main}, each point on such a curve corresponds to an extremal trajectory of the original fixed-initial-mass minimum-time problem. The proposed procedure therefore enables the systematic identification of \deleted{all} time-optimal solutions \deleted{as local minima of $t_f$} along the iso-$\mathcal{M}$ curve, \added{and thus identify the best transfer option among them}. In Fig.~\ref{fig:main}, we illustrate the smooth structure of the mapping $\mathcal{M}(t_f, N_{\mathrm{rev}})$.
Several iso-$\mathcal{M}$ curves are reported, including in particular the level set corresponding to $\mathcal{M} = 1500~\mathrm{kg}$. Along this curve, multiple extremal trajectories are identified, associated with either local minima or maxima of the final time $t_f$. The trajectories are also visualized in Figure \added{\ref{fig:main}}.
Some of these solutions have been previously reported in the literature~\cite{caillau20033d, wang2024thrust}, obtained via different homotopy strategies. In contrast, the present perspective in terms of $\mathcal{M}$ provides a unified and comprehensive representation of the solution space. A summary of all identified trajectories is given in Table~\ref{tab:alltraj}.
We emphasize that $\mathcal{M}(t_f, N_{\mathrm{rev}})$ defines a smooth surface, which enables continuous and well-posed continuation procedures. Compared to previous studies on this benchmark, the maximum-initial-mass framework not only recovers the best time-optimal solution, but also exposes a broader family of extremal trajectories through continuation. In particular, it makes clear that some solutions previously reported in the literature (in particular the one marked with E) correspond to local maxima of \(t_f\) along the iso-\(\mathcal{M}\) curves, rather than to local minima.

\added{Ordering the solutions of Table~\ref{tab:alltraj} by $N_\mathrm{rev}$, local minima and maxima of $t_f$  alternate along the $1500$~kg iso-$\mathcal{M}$ curve: A (min), D (max), B (min), E (max), C (min). 
The minima decrease towards B from both sides ($283.32$~h $\to 281.97$~h $\leftarrow 283.78$~h) and are separated by the maxima D and E, which act as barriers between them (Fig.~1). 
In the terminology of multimodal landscapes~\cite{doye1999evolution}, this is a funnel: a set of local minima whose cost decreases towards the best solution, separated by local maxima. 
We observe this structure in the explored range $N_\mathrm{rev}\in[14.3,16.2]$ and make no claim on its generality. 
The iso-$\mathcal{M}$ curve traverses the funnel directly, visiting each minimum and each separating maximum in turn. By contrast, thrust continuation on the minimum-time problem tracks a single branch, must be regularized to remain well-posed~\cite{wang2024thrust}, and reaches different minima only through separate runs or additional continuation layers. 
On this benchmark, the single continuation used here recovers solutions A, B and E reported in~\cite{wang2024thrust, caillau20033d}, adds the previously unreported extremals C and D, and classifies E as a local maximum of $t_f$.}

\begin{figure*}[tb]
\centering
\includegraphics[width=1.0\textwidth]{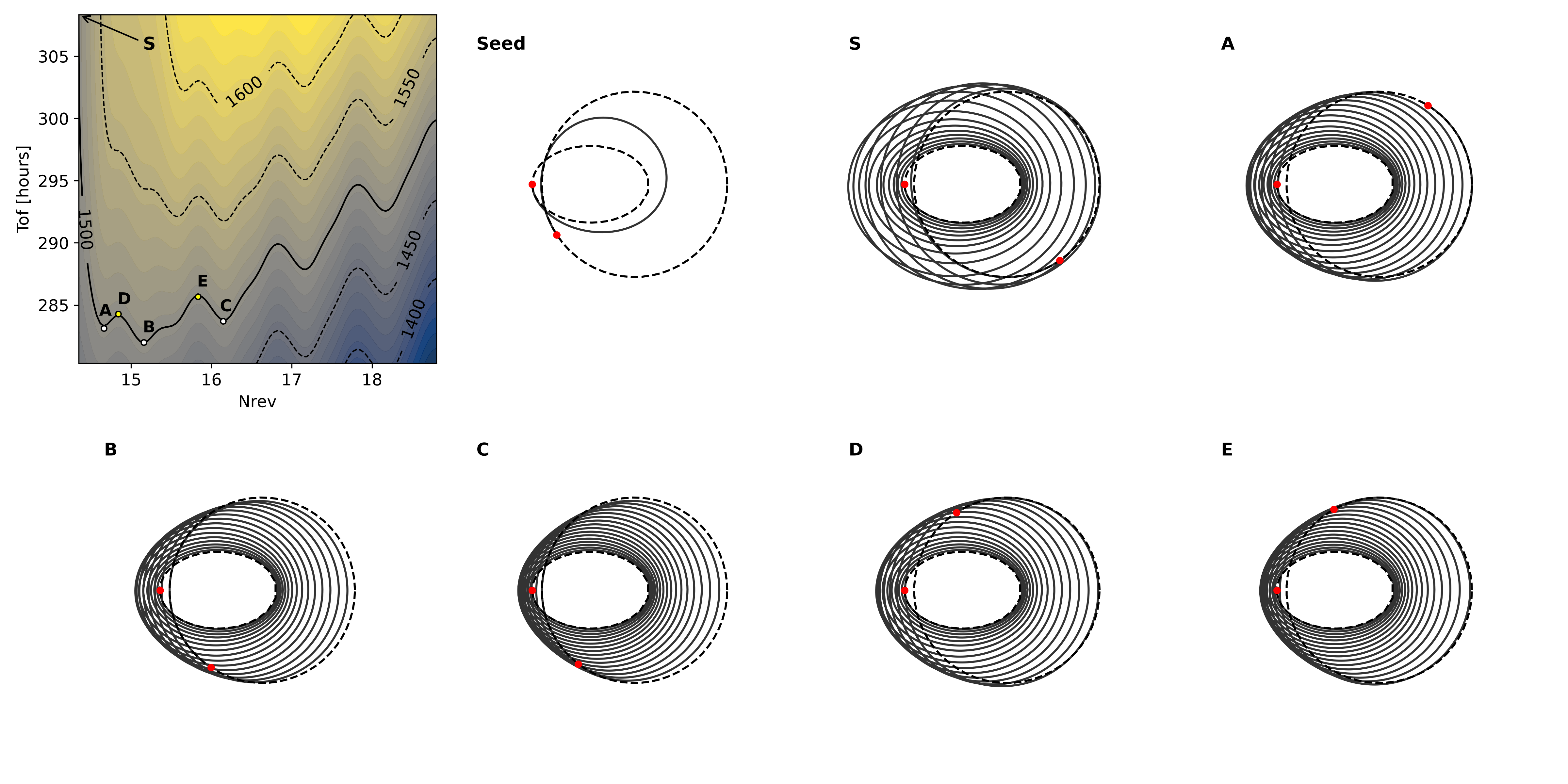}
\caption{\deleted{Relevant t} \added{T}rajectories along the continuation. \deleted{From S included onwards a}\added{A}ll trajectories \deleted{lie on the 1500kg iso-$\mathcal M$ curve}\added{ after the seed S have $m_0=1500$kg}.\label{fig:main}}
\end{figure*}

\section{Conclusion}

We formalized the maximum-initial-mass problem as a standalone optimal control formulation for low-thrust trajectory design and analyzed its structure within Pontryagin’s framework. 
Under standard low-thrust dynamics, the corresponding necessary conditions yield a simple control structure: in the nondegenerate case, the optimal throttle is full-on and the thrust direction follows a primer-vector law, which makes the resulting boundary-value problem well suited to both direct and indirect numerical methods.
We explored the relationship between maximum-initial-mass and minimum-time problems. 
By revisiting the derivation of the terminal Hamiltonian condition, we showed that the commonly used condition on the terminal value of the Hamiltonian in minimum-time problems is a gauge choice rather than a genuine necessary condition \added{suggesting, as future work, the study of alternative gauge conditions to improve the numerical stability of the shooting method employed}. 
\added{Utilizing this remark, we established an extremal correspondence: under the stated assumptions, any minimum-time or maximum-time extremal with fixed initial mass induces a maximum-initial-mass extremal on the same fixed time interval. 
Conversely, a maximum-initial-mass extremal admits a normal minimum-time or maximum-time representation according to the sign of $\mathcal H^{MIM}(t_f)$, or, when $\mathcal H^{MIM}(t_f)=0$, an abnormal representation as both a minimum-time and a maximum-time extremal.}
Despite this correspondence, the maximum-initial-mass and minimum-time, \added{maximum-time} formulations remain distinct problems, and they naturally unlock different algorithmic pathways to explore the low-thrust solution space. \added{In reachability analysis, the maximum-initial-mass formulation combines the two conditions arising from time extrema into a single mass condition.}
Furthermore, the homotopy construction presented here shows that continuation in transfer time within the maximum-initial-mass setting leads to a multi-revolution path that is qualitatively different from thrust continuation on the time-optimal problem, and \added{able to reliably reveal interesting solutions.}\deleted{in practice simpler than,}
\deleted{A similar perspective applies to reachability analysis, which is classically tied to the time-optimal problem and, through the present results, can be profitably recast in terms of maximum-initial-mass level sets. }
Taken together, these observations support the use of the maximum-initial-mass formulation as a complementary design primitive alongside the minimum-time problem, particularly in preliminary mission design and in data-driven approaches where large ensembles of optimal transfers must be generated and approximated efficiently.

\bibliography{sample}

\end{document}